\documentclass[12pt]{article} 

\usepackage{theorem}\usepackage{amsmath,amssymb, bbm}
\usepackage{amsfonts,graphicx,verbatim,psfrag}\usepackage{graphics}
\usepackage{verbatim, hyperref}\usepackage{rotating}\usepackage{epsfig}\usepackage{color}

\usepackage{cleveref}
\usepackage{bm}

\usepackage{natbib}
\newcommand{\be}{\begin{equation}} \newcommand{\ee}{\end{equation}}
\newcommand{\bea}{\begin{eqnarray}} \newcommand{\eea}{\end{eqnarray}}
\newcommand{\beas}{\begin{eqnarray*}} \newcommand{\eeas}{\end{eqnarray*}}
\newcommand{\bit}{\begin{itemize}}
\newcommand{\eit}{\end{itemize}}
\newcommand{\ben}{\begin{enumerate}}
\newcommand{\een}{\end{enumerate}}

\newcommand{\ba}{\begin{array}} \newcommand{\ea}{\end{array}}

\newcommand{\lbar}[1]{\mbox{\b{$#1$}}}

\newcommand{\mc}{\mathcal}

\newcommand{\noi}{\noindent}

\newcommand{\Bthe}{\begin{theorem}}
\newcommand{\Ethe}{\end{theorem}}

\newcommand{\R}{\mathbbm{R}}

{\theoremstyle{plain}} \theoremheaderfont{\scshape}

\newtheorem{theorem}{Theorem}
\newtheorem{corollary}{Corollary}

\newtheorem{proposition}{Proposition}
\newtheorem{lemma}{Lemma}

\newtheorem{assumption}{Assumption}
\newtheorem{definition}{Definition}

\newenvironment{proof}{\noi {\it Proof.} }

\begin{document}

\title{Liability Design with Information Acquisition}
\date{\today}
\author{Francisco Poggi and Bruno Strulovici \\ Northwestern University\thanks{Emails: fpoggi@u.northwestern.edu and b-strulovici@northwestern.edu. Strulovici gratefully acknowledges financial support from the National Science Foundation (Grant No.1151410).}}

\setlength{\parindent}{0pt}
\setlength{\parskip}{.3cm}

\maketitle

\begin{abstract} How to guarantee that firms perform due diligence before launching potentially dangerous products? We study the design of liability rules when (i) limited liability prevents firms from internalizing the full damage they may cause, (ii) penalties are paid only if damage occurs, regardless of the product's inherent riskiness, (iii) firms have private information about their products' riskiness before performing due diligence. We show that (i) any liability mechanism can be implemented by a {\em tariff} that depends only on the evidence acquired by the firm if a damage occurs, not on any initial report by the firm about its private information, (ii) firms that assign a higher prior to product riskiness always perform more due diligence but less than is socially optimal, and (iii) under a simple and intuitive condition, any type-specific launch thresholds can be implemented by a monotonic tariff.
\end{abstract}

\section{Introduction}

In 2019, a California court sentenced paint maker Sherwin-Williams to pay hundreds of millions of dollars to address the dangers caused by lead paint. The sentence was remarkable because even though lead paint became banned in 1978, the suit concerned damage caused during the decades {\em before} the ban and centered on the accusation that paint makers were aware of the dangers caused by lead paint long before the ban was passed.

In essence, the court's argument was that Sherwin-Williams and other paint makers knew or should have known the dangers caused by lead paint.

While it is difficult for a regulator to guess a firm’s private information, it is perhaps easier to assess due diligence: did paint makers research the risks of lead paint sufficiently well before marketing it?

Formally, the problem is not just one of private information, but also one of information acquisition: how can a regulator make sure that agents learn sufficiently well before taking actions?

One may model this question as a delegated Wald problem (\cite{Wald1945}): the principal is a regulator who relies on an agent (the firm) to acquire information before deciding between launching a product and abandoning it.

If the regulator could unrestrictedly penalize a firm, she could force the firm to internalize any damage caused by the product and implement the socially-optimal level of information acquisition.

For various reasons, liability may be capped, however, which precludes the full transfer of damages to the firm. Moreover, the regulator may punish the firm only if some damage occurs, and choose a penalty that depends only on the information available to the regulator after the damage has occurred.

We analyze this problem in a Brownian version of the Wald Problem: the firm observes an arithmetic Brownian whose drift depends on the state of the world, i.e., on the riskiness of the product. Information acquisition is costly. The first-best policy is to acquire information until the riskiness of the product becomes sufficiently clear, launch the product if this riskiness is low, and abandon it if the riskiness is high.

We characterize all incentive-compatible liability rules when (i) the firm has initial private information, (ii) liability is capped, and the (iii) regulator can penalize the firm only when damage occurs.  In general, the regulator may wish to propose at the outset a menu of contracts to the firm in order to extract some of the firm's private information. Indeed, this is the approach suggested by the Revelation Principle. In the present context, however, this approach may be difficult to implement, because it requires that the firm contracts with the regulator long before launching the product and, in fact, even before knowing whether it wishes to launch the product.

Fortunately, our first main result is that it is without loss generality for the regulator to focus on {\em tariff mechanisms}, which are mechanisms for which the firm does not report its private information and only pays a penalty if damage occurs. This result may be viewed as a Taxation Principle for situations in which transfers take place only after some contingencies (damage occurs), but not others, and builds on our companion paper (\cite{PoggiStrulovici2020a}), which provides a general Taxation Principle with Non-Contractible Events.

With a tariff mechanism, a firm's decision to launch the product depends on its prior information, which affects the probability that the product causes damage. Our second main result is that any incentive-compatible tariff mechanism has the following property: firms whose initial private information assigns a higher probability of damage always acquire more evidence before launching their product. This monotonicity property is not an immediate consequence of incentive compatibility, and would in fact be violated if the regulator could impose evidence-based transfers to the firm regardless of whether a damage occurred.

Our third main result is to show that any launch thresholds that induce the firm to perform more due diligence that it would under a fixed penalty can be implemented by a monotonic tariff, i.e., a tariff whose penalty is decreasing in the strength of evidence acquired by the firm before launching the product.

We also show that for a general specification of the regulator's objective function, setting the tariff at its uniform ceiling induces to little due diligence compared to the social optimum, even when the social benefit from launching the product exceeds the firm's profit from doing so. This result holds under a cost-benefit ratio condition, which stipulates that the social benefit from the product relative to the harm it may cause is smaller that the firm's profit relative to the maximum liability that it may face.

\section{Model}
\label{sec:baseline_model}

A firm must decide between launching a product and abandoning its development. If launched, the product may cause damage with positive probability. The firm has some private information about the product's riskiness and can acquire additional information (``due diligence''), before making a final decision.

A regulator wishes to encourage the launch of low-risk products and deter the launch of high-risk ones, as well as to encourage the firm to acquire sufficient information before making its decision.

The regulator faces two constraints. First, the firm has limited liability: the social cost caused by product damage is $L > 0$ and the firm's liability is capped at some lower level $l <L$. Second, the regulator can penalize the firm only if damage occurs. In particular, it cannot penalize firms that acquired too little information and took an overly risky decision unless such risk results in actual damage.

The timing of the game is as follows:
\bit
\vspace{-.3cm}
\item[] 1. The firm is endowed with a prior $\theta\in \Theta \subset [0,1]$ about the product's riskiness $y\in \{0,1\}$, with $\theta = \Pr(y= 1)$.
\item[] 2. The firm can acquire additional information about $y$ according to a dynamic technology to be described shortly.
\item[] 3. The firm decides between launching and abandoning the product.
\item[] 4. If the firm launches the product, it causes some damage if the product was risky $(y=1)$ and doesn't if the product was safe $(y=0)$.
\item[] 5. In case of damage, the firm pays a penalty $\psi \leq l$ set by the regulator.
\eit
The assumption that a risky product causes damage with probability 1 is without loss of generality: if this probability were less than~1, the same analysis would apply using expected damage and expected penalties.

{\bf Information structure:} During the information-acquisition stage, the firms observes a process $X$ given by
\[X_t = (-1 + 2 y) t + \sigma B_t\]
where $B$ is the standard Brownian motion. The drift of $X$ depend symmetrically on the product's riskiness~$y$: the drift is $+1$ if the product causes damage and $-1$ if it does not. Therefore, observing $X$ gradually reveals $y$. This revelation is progressive due to the stochastic component of $X$.

The firm stops acquiring information at some time $\tau$ that is adapted to the filtration of $X$.

The regulator observes nothing about $X$ except if some damage occurs, in which case she observes the last value $X_\tau$ taken by the process at the time of the firm's decision. $X_\tau$ is a measure of the firm's due diligence to assess the product's riskiness before launching it: in this Brownian model, it is well-known (though not immediate) that for each $t>0$, the variable $X_t$ is a sufficient statistic for the information about $y$ contained by the entire path $\{X_s\}_{s\leq t}$ of the process $X$ until time $t$. Mathematically, the likelihood ratio of $y$ associated with a path of $X$ from time 0 to $t$ is only a function of $X_t$.

Because the stopping time $\tau$ is chosen endogenously by the firm, which has private information about $y$, $X_\tau$ is not a sufficient statistic for $y$ once the firm's strategic timing is taken into account. Our assumption that the regulator observes $X_\tau$ instead of the entire path $\{X_t\}_{t\leq \tau}$ captures the idea that the regulator does not perfectly observe all the decisions made by the firm during the information acquisition stage. Intuitively, the regulator observes the most informative signal about $y$ contained by the path of $X$ that is independent of the firm's private information.

{\bf Payoffs:} The firm incurs a running cost $c$ from acquiring information, and a profit $\pi$ if it launches the product.
Let $d = 1$ if the firm launches the product and $d= 0$ if it abandons it, and $\tau$ denote the time spent acquiring information.
The firm's realized payoff is
\[u = d(\pi - y \psi)  - c\tau\]
where $\pi$ is the firm's profit from the launch in the absence of damage.
The regulator's objective internalizes the entire damage caused by the product:
\[v= d(\beta - y L)  - c\tau\]
where $\beta$ is the social benefit from the launch in the absence of damage.

Throughout the paper, we make the following assumption:
\begin{assumption}[Ordered Cost-Benefit Ratios]\label{A-ranked-ratios} $l/\pi < L /\beta$.\end{assumption}
This assumption captures the idea that the risk of damage is more severe for the regulator relative to the benefit of launching the product than it is for the firm. The assumption allows the social benefit from launching the product to exceed the firm's profit (i.e., $\beta >\pi$).

\section{Preliminary Analysis: Symmetric Information}

{\bf First Best:} If the regulator knew the firm's type $\theta$ and could dictate the firm's strategy, the optimal strategy would consist in launching the product if the process $X$ drops below some lower threshold $x^*_\theta$ and abandoning it if $X$ exceeds some upper threshold $\bar x^*_\theta\geq x^*_\theta$.

{\bf Tariffs:} A tariff is a function $\psi: \R \rightarrow \R$ mapping evidence $x$ to a penalty $\psi(x)\leq l$.\footnote{We allow negative tariffs, which amount to a subsidy for the firm and may be used to reward firms that performed unusually careful inspections before launching their products.} Given a tariff $\psi$, a firm with prior $\theta$ chooses a stopping time $\tau$ and a launch/abandonment decision $d\in \{0,1\}$ to maximize its expected utility
\be\label{eq-incentives} E\left[d(\pi - y \psi(X_\tau)) - c\tau \;| \; \theta\right].\ee
It is straightforward to check that the solution to this problem consists of cutoffs $\lbar x^\psi_\theta < \bar x^\psi_\theta$ such that the firm acquires information until $X$ reaches either of the cutoffs.

Limited liability affects incentives in two ways. First, since the firm does not fully internalize damages, it is willing to take riskier decisions than is socially optimal for a given belief about the product's safety. Second, the value of information is different. For example, if the tariff is $\psi \equiv 0$, the firm has no incentive to acquire any information and always launches its product immediately.

To appreciate the consequences of limited liability, suppose that the regulator sets the tariff uniformly equal to the allowed maximum: $\psi(x) \equiv l$ for all $x\in \R$. In this case, the firm launches the product if $X$ drops below some cutoff $\lbar x^l_\theta$ and abandons it if $X$ reaches some upper cutoff $\bar x^l_\theta$.

This maximum penalty may motivate the firm to perform due diligence before launching the product, but the amount of due diligence is always strictly suboptimal, as the next result shows.
\begin{proposition}[recklessness]\label{pro-under} $\lbar x^*_\theta < \lbar x^l_\theta$ for all $\theta\in \Theta$.\end{proposition}
\begin{proof} We fix some prior $\theta \in\Theta$ throughout the proof and let $x^*$ and $x^l$ denote the socially-optimal and firm-optimal launch thresholds, respectively, when $\psi \equiv l$, given prior $\theta$.

Given a current evidence level $x$, the firm's expected payoff if it launches the product at $x$ is:
\[u(x) = \pi - p(x) l\]
where $p(x) = \Pr(y = 1 | x, \theta)$. The regulator's expected payoff if the firm stops at $x$ is:
\[v(x) = \beta - p(x) L.\]
Assumption~\ref{A-ranked-ratios} implies that
\be\label{eq-decomposition} v(x) = \frac{L}{l}(u(x) - k)\ee
where $k = \pi - \beta l/L > 0$.

Thus, the ``launch-payoff functions'' faced by the regulator and the firm are related by equation~\eqref{eq-decomposition}, and both parties face a running cost $c$ before launching or abandoning the product and a payoff normalized to zero if the product is abandoned. Proposition~\ref{pro-under} then follows from two observations:

{\bf Observation 1:} Consider two launch-payoff functions $\hat u,u$. If $\hat u = \alpha u$ with $\alpha >1$, then the optimal launch threshold for $\hat u$ is lower than the optimal launch threshold for $u$.

{\bf Observation 2:} Consider two launch-payoff functions $\hat u,u$. If $\hat u = u - \hat k$ with $\hat k >0$, the optimal launch threshold for $\hat u$ is lower than the optimal launch threshold for $u$.

Once we justify these observations, Proposition~\ref{pro-under} follows from~\eqref{eq-decomposition} by applying Observation 2 to $u-k$ and $u$ and Observation 1 to $v =L/l (u-k)$ and $u-k$, using the fact that $L/l  >1$.

To prove Observation 1, notice that if $\hat u = \alpha u$ with $\alpha >1$, the dynamic optimization problem with launch payoff $\hat u$ and running cost $c$ is equivalent to the problem with launch payoff $u$ and running cost $\hat c = c/\alpha < c$, since the problems become identical up to the scaling factor $\alpha$. With a lower running cost $\hat c$, the continuation interval $(\lbar x(\hat u), \bar x(\hat u))$ contains the continuation interval $(\lbar x(u), \bar x(u))$ with running cost $c$. In particular, the launch thresholds are ranked: $\lbar x(\hat u) \leq \lbar x(u)$.

To prove Observation 2, consider the optimal continuation interval $(x^l, \bar x)$ when the launch-payoff function is $u$ and let $\tau = \inf\{t: X_t\notin (x^l,\bar x)\}$. Fixing any $x\in (x^l, \bar x)$, acquiring information is optimal when starting at $x$, which means that
\be\label{eq-u} u(x) \leq f(x) u(x^l) - c E_x[\tau]\ee
where $f(x)$ is the probability that $X_\tau = x^l$ (as opposed to $\bar x$) and $E_x[\tau]$ is the expected value of $\tau$ when the process $X$ starts at $x$. For the launch-payoff function $\hat u = u - k$ with $k >0$,~\eqref{eq-u} implies that
\[\hat u (x) < f(x) \hat u (x^l) - c E_x[\tau].\]
This shows that stopping at $x$ to launch the product is strictly dominated by the strategy that consists in launching the product if $X$ reaches $x^l$ and abandoning it $X$ reaches $\bar x$. This implies that the optimal launch threshold with $\hat u$ is lower than $x^l$ and proves Observation~2.\hfill$\blacksquare$\end{proof}
Intuitively, Proposition~\ref{pro-under} captures the idea that the regulator values more than the firm having a safer product conditional on launch. Remarkably, however, this result holds even when the social benefit from launching the product exceeds the firm's profit from doing so.

Although the uniform tariff $\psi \equiv l$ brings the firm closest to fully internalizing the damage that its product might cause, the regulator might choose a different tariff, for example, to reward the firm if it acquired more information. The next section studies the firms' incentives in more details.

\section{Incentive Compatibility}

Suppose that the regulator can contract with the firm after the firm has received its initial private information and before it takes any action, and that the regulator has full commitment power.
\begin{definition} A \emph{direct liability mechanism} is a menu $M = (\{\tau_\theta, d_{\theta}, \psi_\theta\}_{\theta\in \Theta})$ such that for all $\theta\in \Theta$:
\bit
\vspace{-.3cm}
\item[(i)] The stopping time $\tau_\theta$ is measurable with respect to the filtration $\{\mc F^X_t\}_{t\geq 0}$ generated by $X$;
\item[(ii)] The decision $d_\theta$ is measurable with respect to the information at time $\tau$, i.e., to the $\sigma$-algebra $\mc F^X_{\tau_\theta}$;
\item[(iii)] The tariff $\psi_\theta: \R \rightarrow \R$ is uniformly bounded above by $l$.
\eit
\end{definition}
Since the regulator has full commitment power, the Revelation Principle guarantees that it is without loss of generality to focus on direct liability mechanisms.

Given a direct liability mechanism, the firm chooses an item $f_{\hat \theta} = (\tau_{\hat \theta}, d_{\hat \theta}, \psi_{\hat \theta})$ from the menu. Faced with the tariff $\psi  = \psi_{\hat \theta}$, the firm chooses a stopping time and a decision to maximizes its expected utility as given by~\eqref{eq-incentives}.
\begin{definition} A direct liability mechanism $M$ is {\em incentive compatible} if for each $\theta\in \Theta$ it is optimal to chooses the item $f_{\theta}$ from $M$ and the strategy $(\tau_\theta, d_\theta)$.\end{definition}
In general, a direct liability mechanism may implement absurd policies: for example, the firm could get a very high reward (i.e., a negative penalty) if it launches the product when $X_t$ is very high (and, hence, the product is very risky). We rule out such a possibility and focus on {\em admissible} mechanisms:
\begin{definition} An IC direct liability mechanism is {\em admissible} if each type $\theta$'s strategy is characterized by thresholds $\lbar x_\theta \leq \bar x_\theta$ such that $\theta$ launches the product if $X_t$ drops below $\lbar x_\theta$ and abandons it if $X_t$ exceeds $\bar x_\theta$.\end{definition}
In practice, it may be difficult for a regulator to contract with the firm at the outset and agree on penalties that depend finely on a firm's private information before it launches a product and, even earlier, before the firm decides how much due diligence to perform before deciding whether to launch its product. It is therefore valuable to determine when a direct liability mechanism can be implemented by a tariff that is independent of the firm's private information.
\begin{definition} A direct liability mechanism is a {\em tariff mechanism} if the tariffs $\{\psi_\theta\}_{\theta\in \Theta}$ are independent of $\theta$.\end{definition}
\begin{theorem}\label{thm:posted_liability} Any admissible direct liability mechanism is outcome-equivalent to a tariff mechanism.\end{theorem}
\begin{proof} Consider any direct liability mechanism $M$ and let $\lbar x_\theta = \lbar x^{\psi_\theta}_\theta$ and $\psi_\theta = \psi_\theta(\lbar x_\theta)$ denote the firm's launch threshold and penalty in case of damage that are implemented under mechanism $M$ when the firm has type $\theta$.

We introduce a ceiling mechanism $\tilde M$ as follows: for each $\theta$, $\tilde \psi_\theta$ gives the maximal penalty $l$ for all $x$ except at $\lbar x_\theta$, where it gives $\psi_\theta$. The ceiling mechanism $\tilde M$ is IC and implements the same thresholds $\lbar x_\theta$, because under $M$ the firm faces the penalty only when it launches the product and higher penalties at other levels can only reduce the incentive to deviate.

If $M$ prescribes the same threshold $\lbar x$ to types $\theta\neq \theta'$, the penalties $\psi_\theta$ and $\psi_\theta'$ must be identical. Otherwise, one type would want to misreport its type and $M$ would not be incentive compatible.

We define the tariff $\psi$ as follows:
\[\psi(\lbar x_\theta) = \psi_\theta\]
for all $\theta \in \Theta$ and
\[\psi (x) = l\]
otherwise.

This tariff is independent of the firm's private information. Moreover, it implements the same launch thresholds as $M$, as is easily checked.\hfill$\blacksquare$
\end{proof}
Theorem~\ref{thm:posted_liability} shows that any admissible liability mechanism can be implemented by a tariff. From now on, we invoke Theorem~\ref{thm:posted_liability} and focus without loss of generality on admissible mechanisms that are implemented by tariffs, hereafter ``admissible tariffs''.

Given any admissible tariff $\psi: x \mapsto \psi(x)$, each type $\theta$ faces a Markovian decision problem in which the state variable at time $t$ is $X_t$. Therefore, there exist thresholds $\lbar x^\psi_\theta\leq \bar{x}^\psi_{\theta}$ such that type $\theta$ stops acquiring information when the process $X$ leaves the interval $(\lbar{x}^\psi_\theta, \bar{x}^\psi_{\theta})$, launches the product at $\lbar x^\psi_\theta$ and abandons it at $\bar x^\psi_\theta$.

Our next result establishes a single-crossing property for the firm.
\begin{lemma}\label{lemma-single-crossing} Consider any admissible tariff $\psi$, level $x$, and type $\theta\in \Theta$. If $\theta$ prefers acquiring information at $x$ to immediately launching the product at $x$, then so does any type $\theta' \geq \theta$.\end{lemma}
\begin{proof} We fix a tariff function $\psi$ and a level $x$, and suppose that $X_t = x$ at some time $t$ that we normalize to 0 for simplicity. Suppose that some type $\theta$ prefers the strategy that consists in launching the product at $\lbar x < x$ and abandoning it at $\bar x > x$, and let $p = \Pr(y = 1 | \theta)$.

If $\theta$ launches the product at $x$, it gets:
\be\label{eq-stop} \pi - p \psi(x).\ee
Let $T^g,f^g$ denote the expected hitting time and the probability of hitting $\lbar x$ if $y = 0$ (the product is good), and $T^b$ and $f^b$ be defined similarly if $y =1$ (the product is damaged). If $\theta$ continues until hitting $\lbar x$ or $\bar x$, its expected payoff is
\be\label{eq-continue}p(f^b(\pi - \psi(\lbar x)) - c T^b) + (1-p) (f^g \times \pi - c T^g).\ee
Comparing~\eqref{eq-stop} and~\eqref{eq-continue}, continuing is optimal if
\be\label{eq-compare}p (f^b(\pi-\psi(\lbar x)) + \psi(x) - c T^b) + (1-p) (f^g \pi- c T^g) \geq \pi.\ee
The left-hand side is a convex combination of two terms: $a = f^b(\pi-\psi(\lbar x)) + \psi(x)- c T^b$ and $b = f^g \pi - c T^g$.
The second term, $b$ is less than $\pi$, because $f^g$ is a probability. Therefore,~\eqref{eq-compare} can hold only if the first term, $a$, is greater than $\pi$.

Rewriting~\eqref{eq-compare}, a firm that assigns probability $p$ to $y=1$ wishes to continue if
\[p(a - b) \geq \pi - b.\]
Since $a > b$, the coefficient of $p$ is strictly positive. This implies that any type that assigns probability $p'> p$ to $y=1$ also prefers the continuation strategy to launching the product immediately at $x$.\hfill$\blacksquare$
\end{proof}

Lemma~\ref{lemma-single-crossing} has the following intuition: If a firm knew that the product were safe, it would optimally launch the product immediately. The return to acquiring more evidence is negative in this case. Given any liability function, if a type wants to acquire more evidence it must be that doing so has a positive return conditional on the product being unsafe. The expected return from acquiring more evidence is thus increasing in the probability that the firm assigns to the product being faulty.

Lemma~\ref{lemma-single-crossing} immediately implies the following monotonicity result:
\begin{proposition}\label{prop-mono-thresholds} For any admissible tariff $\psi$, the launch thresholds $x^{\psi}(\theta)$ are decreasing in~$\theta$.\end{proposition}
This monotonicity result crucially hinges on the fact that the regulator can only charge the firm if it causes some damage. The following example\footnote{This example is partially inspired by the approval mechanisms in \cite{Mcclellan2019} and \cite{HenryOttaviani2019}.} shows that if the regulator can charge the firm even when the product causes no damage, the launch thresholds increase with the type of the firm.
\subsection*{Example: Monotonicity Violation with Damage-Independent Fee}
Suppose that the assumptions of our model are maintained with one exception: if the firm launches its product, the regulator charges the firm a fee $\eta(x)\leq l$ that depends on the evidence $x$ demonstrated when the product is launched, independently of any damage subsequently caused by the product.

We assume that $l>\pi$, so that the regulator can deter the firm from launching the product at any given $x$ by setting $\eta(x) = l$. An admissible direct revelation mechanism specifies, for each type, a launch threshold $\lbar{x}_\theta$ and a fee $\eta_\theta = \eta(\lbar x_\theta)$. Without loss of generality we assume that $\eta(x) = l$ for all $x\notin \{\lbar x_\theta\}$. 

For the sake of this example, we assume for simplicity that the firm perfectly knows its product's riskiness, which means that there are two types of firms: \emph{bad} firms with prior $\theta = 0$ and good firms with prior $\theta=1$.

We construct an IC mechanism for which $\lbar{x}_0 < \lbar x_1$. We start by setting a launch threshold $\lbar x_1<0$ for the bad firm and choose $\eta_1 = \eta(\lbar x_1)$ low enough that (i) a bad firm forced to launch the product at $\lbar x_1$ abandons the product at a threshold $\bar{x}_1>0$ such that $\bar x_1 > |\lbar{x}_1|$ and (ii) this strategy yields a strictly positive expected payoff to the bad firm. Such a construction is always possible by choosing $\lbar x_1$ close enough to 0.

Next, we fix some launch threshold $\lbar{x}_0\in (2 \lbar x_1, \lbar x_1)$ for the good firm and choose $\eta_0 = \eta(\lbar x_0)$ so that a good firm is indifferent between launching the product at $\lbar x_0$ and at $\lbar x_1$. Such a construction is always possible by choosing $\lbar x_0$ close enough to $\lbar x_1$ and $\eta_1$ slightly lower than $\eta_0$. By construction, a good firm is indifferent between the two items of menu $\{(\lbar x_0,\eta_0)$ and $(\lbar x_1,\eta_1)\}$.

To demonstrate incentive compatibility, there remains to show that a bad firm prefers the second item on this menu. Let $\bar x^d$ denote the optimal abandonment threshold of the bad firm if it launches its product at $\lbar x_0$. Suppose first that $\bar{x}^d\leq 0$, This means that the firm prefers to abandon immediately, starting from $X_0=0$. This yields an expected payoff of zero and is dominated by the item $(\lbar x_1, \eta_1)$. Now suppose that $\bar{x}^d>0$. By construction, $\lbar x_1$ is closer to $\lbar x_0$ than it is to $\bar x^d$. Lemma~\ref{lemma-mimic}, then implies that at state $\lbar x_1$, a good firm gets a strictly higher payoff than the bad firm by adopting the strategy $(\lbar x_0, \bar x^d)$. The good firm is by construction indifferent between the strategy $(\lbar{x}_0,\bar{x}_0)$ and stopping immediately at $\lbar x_1$, so the good firm weakly prefers to stop immediately to adopt the strategy $(\lbar{x}_0,\bar{x}^d)$. Moreover, both types get exactly the same payoff if they stop at $\lbar x_1$, since both firms pay the $\eta(x_1)$. Therefore, the strategy $(\lbar x_0,\bar x_d)$ must be strictly worse for the bad firm than stopping at $\lbar x_1$. This shows that the mechanism is incentive compatible for both types of firms, and that the risky firm ($\theta = 1$) launches the product with less evidence than the safe one ($\theta = 0$), since $\lbar{x}_0<\lbar{x}_1$.
\begin{lemma}\label{lemma-mimic}
Consider the strategy that consists in launching the product at $\lbar x$ and abandoning it at $\bar x > \lbar x$, and consider any $x\in (\lbar x, (\lbar x + \bar x)/2))$. If $\eta(\lbar x) <\pi$, the expected payoff from the strategy, starting from $X_0 = x$, is higher for the good firm than for the bad firm.\end{lemma}\begin{proof}
We follow the notation used in the proof of Lemma~\ref{lemma-single-crossing}. Starting from $X_0 =x$, the Optional Sampling Theorem applied to the identity function $X_t\mapsto X_t$ and to type $\theta =1$ implies that \[E[X_\tau | x, \theta = 1] = x + E[\int_0^\tau -1 dt] = x - T^b.\] Expressing the expectation on the left-hand side in terms of hitting probability $f^b$ and rearranging yields:
\[(-1)\cdot T^b = (\lbar{x}-x)f^b + (\bar{x}- x) (1-f^b) = (\bar{x}- x) - f^b(\bar{x}- \lbar{x})\]
Proceeding similarly for type $\theta = 0$, we get:
\[T^g = (\lbar{x}-x)f^g + (\bar{x}- x) (1-f^g) = (\bar{x}- x) - f^g(\bar{x}- \lbar{x})\]
Summing the last two equations yields
\be\label{eq-diff-T}T^{g} - T^{b}=  2(\bar{x}- x) - (f^g + f^b) (\bar{x} - \lbar{x}).\ee
We have $x\in (\lbar x, (\lbar x + \bar x)/2))$, which implies that $2(\bar x - x) < \bar x - \lbar x$ and that $f^g+f^b > 1$.\footnote{For the latter inequality, notice that the drifts of $X_t$ are exact opposite for good and bad firms, so that $\frac{1}{2}(f^g+f^b)$ is the probability that the Brownian process $X_t$ with drift either 1 or -1 with equal probability hits $\lbar x$ before $\bar x$ when starting from $x$. Since $x$ is closer to $\lbar x$ than it is to $\bar x$, this probability is greater than $1/2$, which implies that $f^g+f^b >1$.}
Therefore~\eqref{eq-diff-T} is negative, which shows that $T^b > T^g$. The difference of the good and bad types' expected payoffs is given by:
\begin{equation*}
	(f^g (\pi - \eta)- c T^g) - (f^b (\pi - \eta(\lbar{x})) - c T^b) = \underbrace{[f^g - f^b]}_{>0} \underbrace{(\pi - \eta)}_{>0} +  c \underbrace{[T^b - T^g]}_{>0}.
\end{equation*}
which is strictly positive.\hfill$\blacksquare$\end{proof}

\section{Reducing Recklessness}

Proposition~\ref{pro-under}, shows that the regulator would like to implement lower thresholds than the firm when the firm faces with a uniform penalty, regardless of the firm's private information. The next proposition shows that under these circumstances, it is without loss of generality to focus on tariffs that are nondecreasing functions of $x$, i.e., which impose a lower penalty, the more due diligence is demonstrated by the firm.
\begin{proposition} Suppose that $\Theta$ is finite and consider any thresholds $\{x_\theta\}_{\theta\in \Theta}$ that are (i) decreasing in $\theta$ and (ii) such that $x_\theta\leq \lbar x^{l}_\theta$ for all $\theta\in \Theta$. Then, there exists a non-decreasing, piecewise-constant tariff $\psi$ such that $\lbar x^\psi_\theta = x_\theta$ for all $\theta\in \Theta$.\end{proposition}
\begin{proof} We index the elements of $\Theta$ from the smallest $\theta_1$ to the largest $\theta_{|\Theta|}$ and construct the tariff $\psi$ by moving from large values of $x$ to lower ones. We start by setting $\psi(x) \equiv l$ for all $x\geq x_{\theta_1}$. At $x_{\theta_1}$, we lower the tariff to a level $\psi_1$ that makes $\theta_1$ exactly indifferent between launching the product at $x_{\theta_1}$ and at $\lbar x^l_{\theta_1}$. We keep $\psi$ constant at the level $\psi_1$ for $x\in (x_{\theta_2},x_{\theta_1}]$. Since a firm's launch threshold when it faces a constant tariff $\hat l$ is decreasing in $\hat l$, and since $\psi_1 < l$, we have
\[x_{\theta_1} < \lbar x^l_{\theta_1} \leq \lbar x^{\psi_1}_{\theta_1}\]
where $\lbar x^{\psi_1}_{\theta_1}$ is the launch threshold used by type $\theta_1$ when the tariff is constant and equal to $\psi_1$. This implies that type $\theta_1$ prefers threshold $x_{\theta_1}$ to any level $x\in (x_{\theta_2},x_{\theta_1})$.

At $x_{\theta_2}$, we lower the tariff $\psi$ to a level $\psi_2$ that makes type $\theta_2$ exactly indifferent between launching the product at $x_{\theta_2}$ and at its preferred level $\hat x_2$ among all $x > x_{\theta_2}$, given the tariff $\psi$ constructed so far. By the single-crossing property established in Lemma~\ref{lemma-single-crossing}, this implies that $\theta_1$ prefers $\hat x_2$ to any $x_{\theta_2}$ and, combined with the previous paragraph, that $\theta_1$ prefers $x_{\theta_1}$ to any $x\geq x_{\theta_2}$.

We set $\psi$ equal to $\psi_2$ for all $x\in (x_{\theta_3}, x_{\theta_2}]$. Since $x_{\theta_2} \leq \lbar x^l_{\theta_2} \leq \lbar x^{\psi_2}_{\theta_2}$, type $\theta_2$ prefers $x_{\theta_2}$ to any $x\in(x_{\theta_3},x_{\theta_2})$. Another application of Lemma~\ref{lemma-single-crossing} guarantees that type $\theta_1$ also prefers $x_{\theta_2}$ to any $x\in(x_{\theta_3},x_{\theta_2})$.

Proceeding iteratively, we then lower $\psi$ at $x_{\theta_3}$ to a level $\psi_3$ that makes type $\theta_3$ exactly indifferent between launching the product at $x_{\theta_3}$ and at its preferred level $\hat x_3>x_{\theta_3}$ given the tariff $\psi$ constructed so far. Repeated applications of Lemma~\ref{lemma-single-crossing} guarantee that types $\theta_1,\theta_2$ prefer their respective thresholds $x_{\theta_1}, x_{\theta_2}$ to $x_{\theta_3}$. We extend $\psi$ by setting it constant, equal to $\psi_3$ for all $x\in (x_{\theta_4}, x_{\theta_3}]$. The proof is completed by induction.\hfill$\blacksquare$
\end{proof}

\section{Taxation Principle with Identifiable Information Acquisition}

When an IC mechanism implements distinct thresholds for distinct types, the conclusion of Theorem~\ref{thm:posted_liability} is a corollary of the Taxation Principle with Non-Contractible Events of our companion paper (\cite{PoggiStrulovici2020a}).

According to that paper, a mechanism is {\em identifiable} if satisfies two conditions that we translate into the present setting. Let $A$ denote the set of all possible strategies by the firm. Each element of $A$ consists of a pair $(\tau,d)$, where $\tau$ is a stopping time adapted to the filtration of $X$ and $d$ is measurable with respect to $\mc F^X_\tau$. For any subset $A'$ of $A$, let $X(A')$ denote the set of observable outcomes by the regulator if the firm chooses an action $a\in A'$ and causes some damage.

\begin{definition}
An IC mechanism $M$ is identifiable if there exists a partition $\mc A = \{A_k\}_{k=1}^K$ of $A$ such that
\bit
\vspace{-.3cm}
\item[] (i) $X(A_k)\cap X(A_{k'}) =\emptyset$ for all $k\neq k'$.
\item[](ii) All types $\theta$ who choose an action in $A_k$ under the mechanism choose the same action of $A_k$.
\eit
\end{definition}

\begin{proposition}\label{pro-identifiable}If $M$ implements distinct launch thresholds for all types, then it is identifiable.\end{proposition}
\begin{proof} For each $\theta$, let $A_\theta$ denote the set of firm strategies that use launch threshold $\lbar x_\theta$, and let $A_0 = A \setminus (\cup_{\theta\in \Theta} A_\theta)$. By assumption on $M$, $\lbar x_\theta \neq\lbar x_{\theta'}$  for all $\theta\neq \theta'$. Therefore, $A_\theta$ and $A_\theta'$ are disjoint for all $\theta\neq \theta'$ and $\mc A = \{A_0, A_\theta: \theta\in \Theta\}$ forms a partition of $A$. Condition (ii) is trivially satisfied since for each cell of $\mc A$ there is at most one type taking action in that cell. Moreover Condition (i) is also satisfied by construction of the partition: $X(A_\theta) = \{\lbar x_\theta\}$ for all $\theta\in \Theta$ and, hence, $X(A_\theta)\cap X(A_{\theta'}) = \emptyset$ for all $\theta\neq \theta'$.\hfill$\blacksquare$\end{proof}

\begin{corollary} If an IC mechanism $M$ implements distinct launch threshold for all types, it can be implemented by a tariff mechanism.\end{corollary}
\begin{proof} Proposition~\ref{pro-identifiable} implies that $M$ is identifiable. The result then immediately follows from Theorem~1 in \cite{PoggiStrulovici2020a}\hfill$\blacksquare$\end{proof}

\pagebreak

\bibliographystyle{apalike}
\bibliography{LD_20Dec08_3.bib}

\begin{thebibliography}{}

\bibitem[Henry and Ottaviani, 2019]{HenryOttaviani2019}
Henry, E. and Ottaviani, M. (2019).
\newblock {Research and the Approval Process}.
\newblock {\em American Economic Review}, 109(3):911--955.

\bibitem[McClellan, 2019]{Mcclellan2019}
McClellan, A. (2019).
\newblock {Experimentation and Approval Mechanisms}.

\bibitem[Poggi and Strulovici, 2020]{PoggiStrulovici2020a}
Poggi, F. and Strulovici, B. (2020).
\newblock {A Taxation Principle with Non-Contractible Events}.

\bibitem[Wald, 1945]{Wald1945}
Wald, A. (1945).
\newblock {Sequential Tests of Statistical Hypotheses}.
\newblock {\em The Annals of mathematical Statistics}, 16(2):117--186.

\end{thebibliography}

\end{document}